\documentclass[11pt]{article}

\usepackage{amsthm}
\usepackage{amsmath,amssymb,amsfonts}
\usepackage{mathtools}
\usepackage[a4paper, total={160mm,247mm},
 left=25mm,
 top=25mm]{geometry}

\usepackage{authblk}

\theoremstyle{plain}
\newtheorem{theorem}{Theorem}
\newtheorem{proposition}[theorem]{Proposition}
\newtheorem{corollary}{Corollary}
\newtheorem{lemma}{Lemma}

\theoremstyle{remark}

\newtheorem{remark}{Remark}

\theoremstyle{definition}

\newtheorem{construction}{Construction}

\usepackage[font=small,width=.85\textwidth]{caption}

\begin{document}

\title{On maximal almost balanced non-overlapping codes and non-overlapping codes with restricted run-lengths}

\author{Lidija Stanovnik}
\author{Miha Moškon}
\author{Miha Mraz}
\date{}

\affil{University of Ljubljana, Faculty of Computer and Information Science,\\ Večna pot 113, Ljubljana, Slovenia}

\maketitle

\abstract{\noindent This paper concerns non-overlapping codes, block codes motivated by synchronisation and DNA-based storage applications. Most existing constructions of these codes do not account for the restrictions posed by the physical properties of communication channels. If undesired sequences are not avoided, the system using the encoding may start behaving incorrectly. Hence, we aim to characterise all non-overlapping codes satisfying two additional constraints. For the first constraint, where approximately half of the letters in each word are positive, we derive necessary and sufficient conditions for the code’s non-expandability and improve known bounds on its maximum size. We also determine exact values for the maximum sizes of polarity-balanced non-overlapping codes having small block and alphabet sizes. For the other constraint, where long sequences of consecutive equal symbols lead to undesired behaviour, we derive bounds and constructions of constrained non-overlapping codes. Moreover, we provide constructions of non-overlapping codes that satisfy both constraints and analyse the sizes of the obtained codes.
}


\maketitle

\section{Introduction}
Non-overlapping codes have been studied for sixty years due to their ability to facilitate synchronisation. Since the physical characteristics of the available transmission channel may not be compatible with every code, constrained coding is often used to incorporate the additional requirements. In classical storage and communication systems, balanced codes and codes with restricted run-lengths are used to avoid voltage imbalance and charge accumulation between connected components~\cite{immink:2004}. A more novel application of constrained non-overlapping codes, introduced by Yazdi et al.~\cite{yazdi:2018}, is DNA-based storage, where several constraints from biochemical properties and methods should be considered in addition to code balance and restriction of run-lengths. 

Although many constructions of unconstrained non-overlapping codes exist (see, for example, \cite{levenshtein:1970,chee:2013,blackburn:2015,Barcucci:2016,wang:2022,stanovnik:2024}), only a few methods to obtain constrained codes have been developed by Bilotta et al.~\cite{bilotta:2012}, Levy and Yaakobi~\cite{Levy:2018}, and Yazdi et al.~\cite{yazdi:2018}. None of them attain the theoretical upper bound on the constrained code size. Hence, whether they are maximum cannot be concluded from current knowledge. One way of approaching this problem is by characterising the set of all non-overlapping codes satisfying a specified constraint. Then, the properties of this set can be studied to determine the optimal size. 

This paper studies non-overlapping codes satisfying the balance and the run-length constraint. The first deals with the number of occurrences of positive and negative symbols in a word that should not differ by too much, and the second deals with avoiding long sequences of consecutive equal symbols. For each constraint, we first derive bounds on the code size and then determine the set of non-overlapping codes satisfying the constraint that includes all other such codes as subsets. With a computer search, we also determine the optimal code sizes for non-overlapping codes of even block sizes where the number of positive and negative symbols is equal.

The rest of the paper is organised as follows: we introduce the necessary notation and recall some known results in Section~\ref{sec:preliminaries}.  Section~\ref{section:balanced} characterises the maximal \mbox{$\epsilon$-balanced} non-overlapping codes and determines an upper bound on their size. Section~\ref{sec:rll} deals with non-overlapping codes avoiding long runs, and Section~\ref{sec:balanced_rll} with codes satisfying both constraints. Since our constructions are given in a manner that requires exponential space, we provide an enumerative encoding algorithm in Section~\ref{sec:encoding}. Section~\ref{sec:discussion} concludes the paper and suggests future research directions.

\section{Preliminaries}\label{sec:preliminaries}
\subsection{Notation and definitions}
 This paper will use the following notation. The set of integers $1,\dots, n$ is denoted by $\left[n\right]$, and the set of integers $i, \dots, n$ by $\left[i,n\right]$. We use the Macaulay brackets to denote the ramp function
\[\langle k \rangle \coloneqq \begin{cases}
    0 & \text{if } k < 0\\
    k & \text{otherwise},
\end{cases}\]
and $\mu(d)$ to denote the Möbius function
\begin{align*}
    \mu(n) = \begin{cases}
        (-1)^t & \text{if } n = p_1p_2\cdots p_t \text{ for distinct primes } p_i, \\
        1 & \text{if } n = 1, \\
        0 & \text{if } n \text{ is divisible by a square.}
    \end{cases}
\end{align*}
Recall that a \emph{weak composition} of a positive integer $n$ with $k$ parts is a $k$-tuple of non-negative integers whose sum equals $n$. A \emph{composition} of a positive integer $n$ with $k$ parts is a $k$-tuple of positive integers whose sum equals $n$. The number of all compositions of $n$ into $k$ parts equals $\binom{n-1}{k-1}$, and the number of compositions of $n$ into $k$ parts such that no part exceeds $l$, herein denoted by $c(n,k,l)$, is given by the formula in Eq.~\eqref{eq:cnkl}. For its proof, we refer the reader to \cite{flajolet:2009}.
\begin{align}\label{eq:cnkl}
    c(n,k,l) = &\sum_{j = 0}^{\lfloor (n-k)/l \rfloor} (-1)^j \binom{k}{j} \binom{n-lj-1}{k-1}.
\end{align}

 For sets of words $L$ and $R$, $LR$ denotes the concatenation of sets $L$ and $R$, i.e., the set of all words of the form $lr$, where $l \in L$ and $r \in R$. If $x$ is a word, we write $xR$ and $Lx$ to denote the sets $\{x\}R$ and $L\{x\}$ respectively. $L^i$ denotes the concatenation of the set $L$ with itself $i$ times. The Kleene star denotes the smallest superset that is closed under concatenation and includes the empty set $L^* = \bigcup_{i \geq 0} L^i$, and the Kleene plus denotes the smallest superset that is closed under concatenation and does not include the empty set  $L^+ = \bigcup_{i > 0} L^i$.
A word $w \in \Sigma^n$ has a \emph{period} $i$ if $w = (w_1\cdots w_i)^+$. The smallest number $i$ satisfying this relation is called the \emph{least period}, and words with the least period $n$ are called \emph{primitive}.

A \emph{balanced Dyck word} is a binary word of an even length $2n$ composed of $n$ zeroes and $n$ ones such that none of its prefixes has more zeroes than ones. We denote the set of all balanced Dyck words of length $n$ by $\mathcal{D}_{2n}$. It is well known that $D_{2n} = C_n$, where $C_n$ is the $n$th Catalan number, $C_n = \frac{1}{n+1} \binom{2n}{n}$.
A \emph{Smirnov word} (also called a \emph{Carlitz word}) is any word $w_1\cdots w_n$ with no consecutive equal letters, i.e., $w_i \neq w_{i+1}$ for $i \in \left[n-1\right]$. If, in addition, its cyclic shifts are Smirnov, the word $w_1\cdots w_n$ is called \emph{cyclic Carlitz}. The ordinary generating function of cyclic Carlitz words on $\Sigma$, determined in \cite{MacFie:2019}, is
\begin{align*}
    F_q(z) = q (q-1) \frac{z^2}{(z+1)(1 - (q-1)z)}.
\end{align*}

Whenever referring to a block code in $\Sigma^n$, we assume that the integer $n \geq 2$ and $\Sigma$ is an alphabet of size $q \geq 2$. We partition the alphabet $\Sigma$ into two non-empty sets, the positive letters $\Sigma_P$ and the negative letters ${\Sigma_N}$, and denote the number of positive letters in a word $w \in \Sigma^n$ by $p(w)$. For an $\epsilon$ such that $0 \leq \epsilon \leq 0.5$, a word $w$ is called \emph{$\epsilon$-balanced} if $\lceil (0.5-\epsilon) n \rceil \leq p(w) \leq \lfloor (0.5 + \epsilon) n\rfloor$. A more vague term, \emph{almost balanced}, is often used to describe words with $p(w)$ close to $n/2$, but the exact definition varies between resources. A word that is $0$-balanced is also called \emph{polarity-balanced}. 

\begin{remark}
When constructing codes over the DNA alphabet, one sets ${\Sigma}_P = \{C,G\}$ and ${\Sigma}_N = \{T,A\}$. In the binary case, ${\Sigma}_P = \{1\}$ and ${\Sigma}_N = \{0\}$. 
\end{remark}
\begin{remark}
In the definition of polarity-balanced words, one usually assumes that both $n$ and $q$ are even and $\lvert \Sigma_P \rvert = \lvert \Sigma_N \rvert = q/2$. However, our constructions work for any non-empty partition of $\Sigma$. In this sense, our definition differs from that of Weber et al.~\cite{weber:2013}. For odd values of $q$, they partition the alphabet into three parts: the set of positive symbols, the set of negative symbols, and the set with a zero symbol, so that $\lvert \Sigma_P \rvert = \lvert \Sigma_N \rvert = \lfloor q/2 \rfloor$ and the number of zero symbols in a codeword is not constrained. On the other hand, we assign the zero symbol either to $\Sigma_P$ or $\Sigma_N$.
\end{remark}

A code $C \subseteq \Sigma^n$ is called \emph{non-overlapping} if no proper prefix of a word in $C$ occurs as a suffix of any word in $C$. It is called \emph{$\epsilon$-balanced} (\emph{polarity-balanced}) if all of its codewords are $\epsilon$-balanced (polarity-balanced respectively). A code contained in a larger code of the same class is called \emph{expandable}. Otherwise, it is called \emph{maximal}. A code with the largest number of words among all codes in the same class is called \emph{maximum}. We denote the size of a maximum non-overlapping code by $S(q,n)$.

\subsection{Constructions}
To our knowledge, two constructions of polarity-balanced non-overlapping codes exist. The binary non-overlapping codes based on Dyck paths due to Bilotta et al.~\cite{bilotta:2012}, provided in Construction~\ref{construction:bilotta} (below), are polarity-balanced if their length is even and $(2n)^{-1}$-balanced if their length $n$ is odd. The sizes of the obtained codes equal
    \begin{align*}
        \lvert D_n \rvert = \begin{cases}
            C_m & n = 2m+1, \\
            \sum_{i \in \left[m/2\right]} C_iC_{m-i} & n = 2m+2, m \text{ even}, \\
            \sum_{i \in \left[(m+1)/2\right]} C_iC_{m-i} - C_{(m-1)/2}^2& n = 2m+2, m \text{ odd}. 
        \end{cases}
    \end{align*}

\begin{construction}
    \label{construction:bilotta}\cite{bilotta:2012}
    The code 
    \begin{align*}
        D_{2n+1} = \{1a: a \in \mathcal{D}_{2n}\}
    \end{align*} is a maximal non-overlapping code.
    If $n$ is even, then the code \begin{align*}
        D_{2n+2} = \bigcup_{i\in\left[0, n/2\right]} \{a1b0: a \in \mathcal{D}_{2i}, b \in \mathcal{D}_{2(n-i)}\}
    \end{align*}
    is a maximal non-overlapping code. If $n$ is odd, then the code 
    \begin{align*}
    D_{2n+2} = \bigcup_{i\in \left[0,(n+1)/2\right]} \{a1b0: a \in \mathcal{D}_{2i}, b \in D_{2(n-i)}\} 
                \setminus \{1a'01b'0, a',b' \in \mathcal{D}_{n-1}\}
    \end{align*}
    is a maximal non-overlapping code.
\end{construction}

Another construction of a polarity-balanced non-overlapping code based on the binary construction of unconstrained non-overlapping codes developed independently in \cite{levenshtein:1970} and \cite{chee:2013}, was presented by Levy and Yaakobi~\cite{Levy:2018}. It is stated in Construction~\ref{construction:levy} (below). Note that these codes also contain no zero runs longer than $k$. The authors prove that the size of the obtained code is lower-bounded by $\left(2^k+2k+1-n\right)\binom{n-2k-2}{n/2-2}$ but for many pairs of values of $k$ and $n$ this bound is negative. Thus, we determine an explicit formula to compute the code size in Proposition~\ref{proposition:size_levy}. Moreover, the authors show that a $q$-ary polarity-balanced non-overlapping code where $\lvert {\Sigma}_P \rvert = \lvert {\Sigma}_N \rvert$ has at most $\binom{n}{n/2} \left(\frac{q}{2}\right)^n / n$ codewords.

\begin{construction}
    \label{construction:levy}\cite{Levy:2018}
    Let $n$ and $k$ be integers, $k \leq n/2$. The set $C(n,k)$ that contains all words of the form $0^k1c1$ where $c$ is a binary word with exactly $n/2-2$ ones and no zero runs of length $k$ is a polarity-balanced non-overlapping code.
\end{construction}

\begin{proposition}\label{proposition:size_levy}
    Let $n > 4$ and $1 \leq k \leq n/2$. The number of words in the code $C(n,k)$ obtained from Construction~\ref{construction:levy} equals
    \begin{align*}
 \sum_{r = \lceil \frac{n}{2k}\rceil - 1}^{n/2-k} \sum_{i=-1}^{1} \sum_{j = 0}^{\lfloor (n/2-2k)/r\rfloor}  (-1)^j(2-\lvert i \rvert)\binom{n/2-3}{r+i} \binom{r}{j} \binom{n/2-k-(k-1)j-1}{r-1}.
    \end{align*}
\end{proposition}

\begin{proof}
    The number of words in $C(n,k)$ equals the number of binary words of length $n-k-2$ with exactly $n/2-2$ ones and no zero run of length $k$. Every binary word has one of the following four forms: $0^{\alpha_1}1^{\beta_1}\cdots 0^{\alpha_r}1^{\beta_r}$, $1^{\beta_1}0^{\alpha_1}\cdots 1^{\beta_r}0^{\alpha_r}$, $0^{\alpha_1}1^{\beta_1}\cdots 0^{\alpha_{r-1}}1^{\beta_{r-1}}0^{\alpha_r}$ or $1^{\beta_1}0^{\alpha_1}\cdots 1^{\beta_r}0^{\alpha_r}1^{\beta_{r+1}}$. If there are exactly $n/2-2$ ones in the word, $\beta$ is a composition of $n/2-2$, and $\alpha$ is a composition of $n/2-k$. The restriction of runs is expressed as a restriction of the size of the parts of $\alpha$. Moreover, for $r < \lceil \frac{n}{2k}\rceil - 1$ and $r > n/2-k$ there are no compositions of $n/2-k$ into $r$ parts not exceeding $k-1$. Thus
    \begin{align*}
        \lvert C(n,k) \rvert = \sum_{r = \lceil \frac{n}{2k}\rceil - 1}^{n/2-k} \biggl(2\binom{n/2-3}{r-1} + &\binom{n/2-3}{r} \\  + &\binom{n/2-3}{r-2} \biggr) c(n/2-k,r,k-1) \\
        =\sum_{r = \lceil \frac{n}{2k}\rceil - 1}^{n/2-k} \sum_{i=-1}^{1} (2-\lvert i \rvert) &\binom{n/2-3}{r+i} c(n/2-k,r,k-1),
    \end{align*}
    and the formula follows from Eq.~\eqref{eq:cnkl}.
\end{proof}

In our previous paper on non-overlapping codes, \cite{stanovnik:2024}, we showed that every maximal non-overlapping code is contained in the set $\mathcal{M}_{q,n}$ (defined below in Construction~\ref{construction:fimmel}). Every non-overlapping code is, as implied by Zorn's lemma, a subset of some code from $\mathcal{M}_{q,n}$ and in particular, this also holds for every $\epsilon$-balanced non-overlapping code and every non-overlapping with restricted run-lengths. Hence, we will study the largest subsets of codes in $\mathcal{M}_{q,n}$ that satisfy the required constraints.

\begin{construction}
    \label{construction:fimmel}\cite{stanovnik:2024}
    Let $n \geq 3$. We define $\mathcal{M}_{q,n}$ to be the set of all codes \[C = \bigcup_{i\in\left[n-1\right]} L_i R_{n-i} \subseteq \Sigma^n,\] where $(L_1, R_1)$ a partition of $\Sigma$ into two non-empty parts, and $(L_i, R_i)$ a partition of $\bigcup_{j\in\left[i-1\right]} L_j R_{i-j} $ for every $i \in \left[2,n-1\right]$.
    Every code $C \in \mathcal{M}_{q,n}$ is non-overlapping.
\end{construction}

To prove some statements regarding the constrained non-overlapping codes, we will need some more properties of the codes from Construction~\ref{construction:fimmel}.

\begin{proposition}
  \label{proposition:z}
  \cite{stanovnik:2024}
    Let $C=\bigcup_{i\in\left[n-1\right]} L_i R_{n-i}$ be a code from Construction~\ref{construction:fimmel} and $X \coloneqq \bigcup_{i\in\left[2n+1\right]} X_i$, such that $X_{2i-1} = L_i$, $X_{2i} = R_i$, and $X_{2n+1} = C$. Then no proper prefix in $X$ occurs as a proper suffix in $X$.
\end{proposition}

\begin{proposition}\label{proposition:unique_index}
  \cite{stanovnik:2024}
  Let $C = \bigcup_{i\in\left[n-1\right]} L_i R_{n-i}$ be a code obtained from Construction~\ref{construction:fimmel}.
  For every word $w \in C$, there exists a unique index $i$ such that $w \in L_iR_{n-i}$.
\end{proposition}

\section{Balanced non-overlapping codes}\label{section:balanced}
Before we proceed with constructing maximal $\epsilon$-balanced non-overlapping codes, we provide some bounds on the size of $\epsilon$-balanced non-overlapping codes. In particular, we first count the number of $\epsilon$-balanced primitive words. This value is then used in Proposition~\ref{balanced:upper_bound} to determine an upper bound on the size of $\epsilon$-balanced non-overlapping codes. Then, we determine the exact size of the maximum polarity-balanced non-overlapping code of length $2$ in Proposition~\ref{proposition:balanced_size_two}.

\begin{proposition}
    The number of $\epsilon$-balanced primitive words in $\Sigma^n$ equals
    \begin{align*}
        \sum_{\substack{d\mid n \\ (0.5-\epsilon)d \leq m \leq (0.5+\epsilon)d}} \mu(d) \binom{d}{m} \lvert {\Sigma}_P\rvert^{m} \lvert {\Sigma}_N \rvert^{d-m}.
    \end{align*}
\end{proposition}

\begin{proof}
    Let $p_1,\dots,p_k$ be the distinct prime divisors of $n$. For an integer $d$, define
    \begin{align*}
        Q_d = \{w \mid w = u^d, (0.5-\epsilon)d \leq p(w) \leq (0.5+\epsilon)d\},
    \end{align*}
    i.e, the set of $\epsilon$-balanced words with a period $d$. Its size equals
    \begin{align*}
        \lvert Q_d \rvert = \sum_{(0.5-\epsilon)d \leq m \leq (0.5+\epsilon)d} \binom{d}{m} \lvert {\Sigma}_P\rvert^{m} \lvert {\Sigma}_N \rvert^{d-m},
    \end{align*}
    since for $m$ between $ \lceil (0.5 - \epsilon)d\rceil$ and $\lfloor (0.5+\epsilon)d\rfloor$, there are $\binom{d}{m}$ choices for the position of positive symbols. There are $\lvert \Sigma_P \rvert^{m}$ choices to select the symbols at these positions, and $\lvert \Sigma_N \rvert^{d-m}$ choices to select the negative symbols in the word.
    Clearly, the number of $\epsilon$-balanced primitive words $B$ equals
    \begin{align*}
        B = Q_1 \setminus \bigcup_{\substack{d \mid n \\ d > 1}} Q_d.
    \end{align*}
    If $e \mid d$, then $Q_e \subseteq Q_d$. Hence, using the principle of inclusion and exclusion and denoting $d(Y) \coloneqq \prod_{i \in Y} p_i$,
    \begin{align*}
        \lvert B \rvert &= \lvert Q_1 \rvert - \sum_{\emptyset \neq Y \subseteq \left[k\right]} (-1)^{\lvert Y \rvert -1} \lvert Q_{d(Y)} \rvert \\
        &= \lvert Q_1 \rvert + \sum_{\emptyset \neq Y \subseteq \left[k\right]} (-1)^{\lvert Y \rvert} \lvert Q_{d(Y)} \rvert\\
        &= \sum_{d\mid n} \mu(d) \lvert Q_d \rvert,
    \end{align*}
    proving the claim.
\end{proof}

\begin{proposition}\label{balanced:upper_bound}
    Let $C$ be an $\epsilon$-balanced $q$-ary non-overlapping code of length $n$. Then
    \begin{align*}
        \lvert C \rvert \leq \sum_{\substack{d\mid n \\ (0.5-\epsilon)d \leq m \leq (0.5+\epsilon)d}} \frac{\mu(d)}{n} \binom{d}{m} \lvert {\Sigma}_P\rvert^{m} \lvert {\Sigma}_N \rvert^{d-m}.
    \end{align*}
\end{proposition}
\begin{proof}
    Let $w$ be a word in $C$. Shifting preserves the number of positive letters in a word, so any cyclic shift of $w$ is $\epsilon$-balanced.
    Since $C$ is a non-overlapping code, none of the cyclic shifts of $w$ is a word in $C$ (otherwise, $w$ has a prefix that occurs as a suffix in $C$). Moreover, $w$ and all of its shifts are primitive -- if any of them has a period $d$, then the prefix of $w$ of length $d$ occurs as a suffix of $w$ of length $d$. Thus, $n\lvert C \rvert \leq \lvert B \rvert$ and the proposition follows. 
\end{proof}

\begin{proposition}\label{proposition:balanced_size_two}
    $S_B(\lvert \Sigma_P \rvert,\lvert \Sigma_N\rvert,2) = \lvert \Sigma_P \rvert\lvert \Sigma_N\rvert.$
\end{proposition}

\begin{proof}
Every polarity-balanced code of length two is obtained by concatenating a positive and a negative letter. A letter that occurs as a prefix should not occur as a suffix to obtain a non-overlapping code.
Hence, we partition the set $\Sigma_P$ into two parts, $L_P$ and $R_P$, and the set $\Sigma_N$ into parts $L_N$ and $R_N$, and write down a polarity-balanced non-overlapping code $C$ in the form $C = L_N R_P \cup L_P R_N$.
Clearly, $\max\; \lvert C \rvert = S_B(\lvert \Sigma_P \rvert,\lvert \Sigma_N\rvert,2)$. Compute
\begin{align*}
    \lvert C \rvert =& \lvert L_N \rvert \lvert R_P \rvert + \lvert L_P \rvert \lvert R_N \rvert \\
    =& \lvert L_N \rvert (q/2 - \lvert L_P \rvert) + \lvert L_P \rvert (q/2 - \lvert L_N \rvert) \\
    =& q/2 (\lvert L_N\rvert + \lvert L_P \rvert) - 2 \lvert L_N \rvert \lvert L_P \rvert.
    \end{align*}
    This function is continuous; therefore, its global maxima lie at the boundaries or critical points.
    The only critical point is $(\lvert L_N\rvert, \lvert L_P \rvert) = (\lvert \Sigma_P\rvert/2,\lvert \Sigma_N\rvert/2)$, but the corresponding Hessian, $H(\lvert \Sigma_P\rvert/2,\lvert \Sigma_N\rvert/2) = \begin{bmatrix}0 & -2 \\ -2 & 0\end{bmatrix}$, has a negative determinant and hence $(\lvert \Sigma_P\rvert/2,\lvert \Sigma_N\rvert/2)$ is a saddle point. The maxima, therefore, lies at the boundaries.
    $C$ is symmetric in $\lvert L_N\rvert$ and $\lvert L_P\rvert$, and therefore we only have to test the cases $\lvert L_N\rvert = 0$ and $\lvert L_N\rvert = 1$.
    If $\lvert L_N\rvert = 0$,
    \begin{align*}
        \lvert C \rvert = \lvert \Sigma_N \rvert \lvert L_P \rvert,
    \end{align*}
    and the maximum, $\lvert \Sigma_N\rvert\lvert \Sigma_P \rvert$, is attained at $\lvert L_P \rvert = \lvert \Sigma_P \rvert$.
    If $\lvert L_N\rvert = 1$,
    \begin{align*}
        \lvert C \rvert =&\, \lvert \Sigma_P \rvert + \lvert \Sigma_N \rvert \lvert L_P \rvert - 2 \lvert L_P \rvert \\
        =&\, \lvert \Sigma_P \rvert + (\lvert \Sigma_N \rvert-2)\lvert L_P \rvert.
    \end{align*}
    For $\lvert \Sigma_N \rvert > 2$, $(\lvert \Sigma_N \rvert-2) > 0$, and the function above attains its maximum at $\lvert L_P \rvert = \lvert \Sigma_P \rvert$.
    \begin{align*}
        \lvert C \rvert \leq&\, \lvert \Sigma_P \rvert + (\lvert \Sigma_N \rvert-2) \lvert \Sigma_P \rvert \\
        =&\, \lvert \Sigma_P \rvert (\lvert \Sigma_N \rvert-1) \\<&\, \lvert \Sigma_P \rvert \lvert \Sigma_N \rvert.
    \end{align*}
    If $\lvert \Sigma_N \rvert = 2$, $\lvert C \rvert = \lvert \Sigma_P \rvert \leq \lvert \Sigma_P \rvert \lvert \Sigma_N \rvert$, and if $\lvert \Sigma_N \rvert = 1$, the maximum is attained at $\lvert L_P \rvert = 0$, hence $\lvert C \rvert = \lvert \Sigma_P \rvert \leq \lvert \Sigma_P \rvert \lvert \Sigma_N \rvert$.
\end{proof}

One can construct an $\epsilon$-balanced non-overlapping code and in particular all maximal $\epsilon$-balanced non-overlapping codes by extracting codewords from $X \in \mathcal{M}_{q,n}$ that have between $\lceil (0.5-\epsilon)n \rceil$ and $\lfloor (0.5+\epsilon)n \rfloor$ positive letters. To obtain an extraction, we will partition the sets of prefixes so that $L_i^j$ contains all words from $L_i$ with $j$ positive letters. Similarly, we partition the sets of suffixes so that $R_i^j$ contains all words from $R_i$ with $j$ positive letters. Some properties of the sets $L_i^j$ and $R_i^j$ are given in Proposition~\ref{proposition:lijrijproperties}.

\begin{proposition}\label{proposition:lijrijproperties}
    (i) If $j > i$, then $L_i^j \cup R_i^j = \emptyset$. \\
    (ii) If $j \neq k$, then $(L_i^j  \cup R_i^j) \cap ( L_i^k \cup R_i^k) = \emptyset$. \\
    (iii) If $i > 1$ and $0 \leq j \leq i$, then  \[L_i^j \cup R_i^j =\bigcup_{k=1}^{i-1} \bigcup_{m=\langle k+j-i \rangle} ^{\min(j,k)} L_k^mR_{i-k}^{j-m}.\]
\end{proposition}
\begin{proof}
    (i) A word $x \in L_i \cup R_i$ has length $i$ and therefore contains at most $i$ positive symbols, so $p(x) \leq i$. \\
    (ii) Let $x \in L_i^j \cup R_i^j$. Then, by definition, $p(x) = j$. Therefore $x \in L_i^k \cup R_i^k$, if and only if $k = p(x) = j$.\\
    (iii) By definition the set $L_i^j \cup R_i^j$ contains all words $w \in \bigcup_{k=1}^{i-1}(L_k R_{i-k})$ such that $p(w) = j$.
    Let $w = uv \in L_i^j \cup R_i^j$ be a decomposition of $w$ such that $u \in L_k$ and $v \in R_{i-k}$. We know by Proposition~\ref{proposition:unique_index} that this decomposition is unique.
    There are at most $j$ positive letters in the prefix $u$ and at most $j$ positive letters in the suffix $v$. After applying statement (i) we obtain
    \begin{align*}
        0 \leq p(u) &\leq \min(k,j), \\
        0 \leq p(v) &\leq \min(i-k,j), \\
        p(u) &+ p(v) = j.
    \end{align*}
    Therefore $0 \leq p(v) = j-p(u) \leq i-k$ and $p(u) \geq j+k-i$.
\end{proof}

Every non-overlapping code from $\mathcal{M}_{q,n}$ can be composed in terms of the new sets (see Proposition~\ref{proposition:mqntolijrij} below). The parameter $k$ in the formula denotes the number of positive letters in a codeword. The $\epsilon$-balanced codewords in $X$ are therefore exactly those with $k$ between $\lceil (0.5 - \epsilon) n \rceil$ and $\lfloor (0.5 + \epsilon) n \rfloor$. This observation leads to Proposition~\ref{proposition:mqntobqn}. Moreover, we determine the number of $\epsilon$-balanced codewords in $X$ in Proposition~\ref{proposition:size}.

\begin{proposition}\label{proposition:mqntolijrij}
    \begin{align*}    
    \bigcup_{i=1}^{n-1} L_i R_{n-i} =
    \bigcup_{i = 1}^{n-1}
    \bigcup_{j=0}^{i}
    \bigcup_{k = j}^{n+j-i} L_i^j R_{n-i}^{k-j}.
    \end{align*}
\end{proposition}

\begin{proof}
From statement (i) of Proposition~\ref{proposition:lijrijproperties} we know that
\begin{align*}
    \bigcup_{i=1}^{n-1} L_i R_{n-i} =
    \bigcup_{ i = 1}^{n-1}
    \bigcup_{j=0}^{i}
    \bigcup_{k = 0}^{n-i} L_i^jR_{n-i}^{k}.
\end{align*}
After rearrangement of the right-hand side by the amount of positive letters in $L_i^j R_{n-i}^k$, $j+k$, we obtain
\begin{align*}
    \bigcup_{i=1}^{n-1} L_i R_{n-i} =
    \bigcup_{i = 1}^{n-1}
    \bigcup_{j=0}^{i}
    \bigcup_{k = j}^{n+j-i} L_i^j R_{n-i}^{k-j}. \mbox{\tag*{\qedhere}}
\end{align*}
\end{proof}

\begin{construction} \label{proposition:mqntobqn}
    Let $X = \bigcup_{i=1}^{n-1} L_iR_{n-i} \in \mathcal{M}_{q,n}$. \\
    The largest $\epsilon$-balanced non-overlapping code contained in $X$ is
    \begin{align*}
    C = 
      \bigcup_{i = 1}^{n-1}
      \bigcup_{k = k_{\text{min}}}^{k_{\text{max}}}
      \bigcup_{j = \langle k + i - n \rangle}^{\min(i,k)}
      L_i^j R_{n - i}^{k-j},
    \end{align*}
    where $k_{\text{min}} = \lceil (0.5 - \epsilon) n \rceil$ and $k_{\text{max}} = \lfloor (0.5 + \epsilon) n \rfloor$.
\end{construction}

\begin{proof}
The set of $\epsilon$-balanced codewords in $X$ is exactly
    \[\bigcup_{i = 1}^{n-1}
    \bigcup_{k = k_{\text{min}}}^{k_{\text{max}}}
      \bigcup_{j = 0}^{\min(i,k)}
      L_i^j R_{n - i}^{k-j}.\]
      If $0 \leq j < k + i - n$, then $k - j > n-i$ and there is no such word $w \in R_{n-i}$ that satisfies $p(w) = k - j$.
      Therefore $R_{n-i}^{k-j}$ is empty.
      $C$ is non-overlapping since every subset of a non-overlapping code is itself non-overlapping.
\end{proof}

\begin{proposition}\label{proposition:size}
Let $C$ be an $\epsilon$-balanced non-overlapping code from Proposition~\ref{proposition:mqntobqn}. Then
    \begin{align*}
        \lvert C \rvert = \sum_{i=1}^{n-1}
            \sum_{k = k_{\text{min}}}^{k_{\text{max}}}
            \sum_{j=\langle i-k \rangle}^{\min(i, k)} \lvert L_i^j \rvert \lvert R_{n-i}^{k -j}\rvert.
    \end{align*}
\end{proposition}

\begin{proof}
    Suppose $w \in L_i^j R_{n-i}^{p-j} \cap L_l^m R_{n-l}^{r-m}$.
    There are $p(w) = p = r$ positive letters in $w$.
    The code $C$ is a subset of some non-overlapping code $X = \bigcup_{i<n} L_iR_{n-i}$ with $L_i^j \subseteq L_i$, $R_{n-i}^{p-j} \subseteq R_{n-i}$, $L_l^m \subseteq L_l$ and $R_{n-l}^{r-m} \subseteq R_{n-l}$.
    Therefore $w\in L_iR_{n-i} \cap L_l R_{n-l}$ and from Proposition~\ref{proposition:unique_index} we get $i = l$.
    Statement (ii) of Proposition~\ref{proposition:lijrijproperties} now yields $j = m$.
\end{proof}

We performed a computer search to determine the maximum sizes of binary polarity-balanced non-overlapping codes of small block sizes. Table~\ref{tab:balanced_binary} shows that Construction~\ref{construction:bilotta} is optimal for many block sizes.  The only parameter value where the sizes differ is $n=14$. On the other hand, the codes obtained from Construction~\ref{construction:levy} are much smaller. Moreover, the computer search revealed that for $q \in \{4,6\}$ and $n \in \{4,6,8\}$ the maximum sizes of polarity-balanced non-overlapping codes also coincide with the sizes of codes obtained by the straight-forward generalisation of Construction~\ref{construction:bilotta}. A computer search of Construction~\ref{proposition:mqntobqn} becomes infeasible for larger parameter values.

\begin{table}[ht]\small
    \centering
    \begin{tabular}{rrrr}
    \hline
         $n$ & Construction~\ref{construction:bilotta} & Construction~\ref{construction:levy}&  Construction~\ref{proposition:mqntobqn}\\\hline
         4 & 1 & 1 & 1\\
         6 & 3 & 2 & 3\\
         8 & 8 & 3 & 8\\
        10 & 23 & 10 & 23 \\
        12 & 72 & 30 & 72 \\
        14 & 227 & 90 & 229 \\
        16 & 760 & 266 & 760 \\\hline
    \end{tabular}
    \caption{\small Comparison of the maximum size of a binary polarity-balanced non-overlapping code and the sizes of codes obtained by Construction~\ref{construction:bilotta} and \ref{construction:levy}.}
    \label{tab:balanced_binary}
\end{table}

One wonders whether it is possible to determine which partitions yield maximal codes. 
We establish necessary and sufficient conditions in Proposition~\ref{proposition:maximal_epsilon_balanced}. This generalises a result from \cite{stanovnik:2024}, where a characterisation of unconstrained maximal non-overlapping codes is provided in a slightly different form since for $\epsilon = 0.5$ the set $R_{n-\sum_{m=0}^{k}i_m}^{p-\sum_{m=0}^{k}j_m} \cup L_{n-\sum_{m=0}^{k} i_m}^{p-\sum_{m=0}^{k}j_m}$ can only be empty when $k=0$, $i_0=n/2$, $j_0 = p/2$ and $L_{i_0}^{j_0} \cup R_{i_0}^{j_0} = \{x_0\}$.

\begin{proposition}\label{proposition:maximal_epsilon_balanced}
An $\epsilon$-balanced non-overlapping code
    \begin{align*}
    C = 
      \bigcup_{i = 1}^{n-1}
      \bigcup_{k = \lceil (0.5 - \epsilon) n \rceil}^{\lfloor (0.5 + \epsilon) n \rfloor}
      \bigcup_{j = \langle k + i - n \rangle}^{\min(i,k)}
      L_i^j R_{n - i}^{k-j},
    \end{align*}
 is maximal if and only if $x_0 \in L_{i_0}^{j_0}$ is not a prefix in $C$ (respectively $x_0 \in R_{i_0}^{j_0}$ not a suffix in $C$) implies that for every $k \geq 0$, $i_0 + \dots + i_k < n$ and $j_0 + \dots + j_k \leq p $ where  $\lceil (0.5-\epsilon)n \rceil \leq p \leq \lfloor (0.5+\epsilon)n \rfloor$
and there exists a sequence of words $(x_1,\dots,x_k)$ that are not prefixes (suffixes) in $C$ such that $x_l \in L_{i_l}^{j_l} x_{l-1}$ for $1\leq l \leq k$ (respectively $x_l \in x_{l-1} R_{i_l}^{j_l}$), either
\begin{align*}
 R_{n-\sum_{m=0}^{k}i_m}^{p-\sum_{m=0}^{k}j_m} \cup L_{n-\sum_{m=0}^{k} i_m}^{p-\sum_{m=0}^{k}j_m} = \emptyset,
\end{align*}
or $k = 0$, $i_0 = n/2$, $j_0=p/2$ and $ L_{i_0}^{j_0} \cup R_{i_0}^{j_0} = \{x_0\}$.
\end{proposition}
\begin{proof}
    $(\Leftarrow):$ \\
    Let $C$ satisfy the assumptions.
    Define $\{p_k(w)\}$ to be a sequence of decompositions of a polarity-balanced word $w$ encoded with a word over the alphabet \\$\{l_0,\dots,l_{\lfloor (0.5+ \epsilon)n\rfloor}, r_0,\dots,r_{\lfloor (0.5 + \epsilon)n\rfloor}\}$ such that \\
    
    \noindent a) $p_0(w) \in \{l_0,l_1,r_0,r_1\}$ with $p_0(w) = l_0$ if $w_i \in L_1^0$, $p_0(w) = l_1$ if $w_i \in L_1^1$, $p_0(w) = r_0$ if $w_i \in R_1^0$ and $p_0(w) = r_1$ if $w_i \in R_1^1$, and \\
    
    \noindent b) if $p_k(w)$ contains $l_{i}r_{j}$ as substring, obtain $p_{k+1}(w)$ by replacing its occurrence with $l_{i+j}$ if $l_ir_j \in L_{m}^{i+j}$ or with $r_{i+j}$ if $l_ir_j \in R_m^{i+j}$ for some $m$.\\
    
    The length of the decomposition is strictly decreasing, so the sequence is finite. The last element of this sequence is therefore in one of the following forms $l_ir_{p-j}$, $l_{\alpha_1}\cdots l_{\alpha_k}$, $r_{\alpha_1}\cdots r_{\alpha_k}$ or $r_{\alpha_1}\cdots r_{\alpha_m}l_{\alpha_{m+1}}\cdots l_{\alpha_k}$ where $(\alpha_1,\dots,\alpha_k)$ is a weak composition of $p \in \{\lceil (0.5+\epsilon)n \rceil, \dots, \lfloor (0.5-\epsilon)n \rfloor\}$.
    If $p(w) = l_ir_{p-j}$, then $w \in C$.
    If $p(w) = l_{\alpha_1}\cdots l_{\alpha_k}$, then by assumption $w$ has a suffix that occurs as a prefix in $C$.
    If $p(w) = r_{\alpha_1}\cdots r_{\alpha_k}$, then by assumption $w$ has a prefix that occurs as a suffix in $C$.
    If $p(w) = r_{\alpha_1}\cdots r_{\alpha_m}l_{\alpha_{m+1}}\cdots l_{\alpha_k}$, then a shift of $w$ that corresponds to $l_{\alpha_{m+1}}\cdots l_{\alpha_k}r_{\alpha_1}\cdots r_{\alpha_m}$ is a word in $C$. Hence, every polarity-balanced word not already in $C$ overlaps with a word in $C$, so $C$ is maximal. \\

    $(\Rightarrow):$ \\
    Suppose $C$ is maximal and $x_o \in L_{i_0}^{j_0}$ is not a prefix in $C$. The case when $x_0 \in R_{i_0}^{j_0}$ is not a suffix in $C$ is symmetric.
    For sake of contradiction, take the smallest $k$ such that there exist $p$ and $(x_1,\dots,x_k)$ that are not prefixes in $C$ but there exists $y \in R_{n-\sum_{m=0}^{k}i_m}^{p-\sum_{m=0}^{k}j_m} \cup L_{n-\sum_{m=0}^{k} i_m}^{p-\sum_{m=0}^{k}j_m}$.
    If $y \in R_{n-\sum_{m=0}^{k}i_m}^{p-\sum_{m=0}^{k}j_m}$, then $x_ky \not\in C$ since $x_k$ is not a prefix in $C$. By Proposition~\ref{proposition:z}, $y$ and none of its suffixes is a prefix in $C$, and no prefix of $y$ occurs as a suffix in $C$.
    If a prefix of $x_k$ or $x_k$ itself is a suffix in $C$, then there is a word or a prefix in $L_{i_m}$ that occurs as a suffix in $C$ violating Proposition~\ref{proposition:z}. Moreover, if $x_ky$ has a suffix longer than $n - \sum_{m=0}^{k}i_m$ that occurs as a prefix in $C$, then either $x_ {i_m}$ is a prefix in $C$ or a word in $L_{i_m}$ has a suffix that is a prefix in $C$. The first cannot occur due to the assumption at the start of this proof, and the second due to Proposition~\ref{proposition:z}.
    Hence, $C \cup \{x_ky\}$ is non-overlapping and larger than $C$. $C \cup \{x_ky\}$ is also $\epsilon$-balanced since $p(x_ky) = p(x_k)+p(y) = p$, violating the assumption that $C$ is maximal.
    If $y \in L_{n-\sum_{m=0}^{k}i_m}^{p-\sum_{m=0}^{k}j_m}$, then we have to prove that unless $k = 0$, $i_0 = n/2$, $j_0=n/4$ and $ L_i^j \cup R_i^j = \{x_0\}$, $yx_k \cup C$ is a non-overlapping code larger than $C$. Note that in the special case mentioned, $yx_k = x_0x_0$ and overlaps with itself. To show that $yx_k\cup C$ is non-overlapping and $\epsilon$-balanced follow a similar procedure as in the previous case. To see that $yx_k$ does not belong to $C$, it is sufficient to observe that $yx_k$ ends in $x_0$ and from Proposition~\ref{proposition:z} we know that $x_0$ is not a suffix in $C$.
\end{proof}

\section{Restricting the length of the maximal run}\label{sec:rll}
Since the first and the last letters of a non-overlapping code are always distinct, the longest run in a cyclic shift of a codeword $w$ is upper-bounded by the longest run of $w$. Furthermore, a non-overlapping code of length $n$ cannot have any run of length $l$ whenever $l \geq n$. Now, denote the number of a maximum $q$-ary non-overlapping code with run-length of at most $l$ and length $n$ by $S_l(q,n)$. We will use the number of cyclically run-length-restricted $q$-ary words to determine an upper bound on the size of a non-overlapping code with restrictions on its run-lengths. We denote the set of $q$-ary words where the longest (cyclic) run does not exceed $l$ by $L_{q,l}^n$ and the corresponding generating function by $L_{q,l}(x)$.

\begin{proposition}
The number of $q$-ary words of length $n$ such that the length of no run in the word nor in any of its cyclic shifts exceeds $l$ equals
\begin{align*}
    \lvert L_{q,l}^n\rvert = \sum_{\sum ik_i = n} \binom{\sum k_i}{k_1,\dots,k_{l}} (q-1) (q^3-q^2-q)^{\sum k_i}.
\end{align*}
\end{proposition}

\begin{proof}
First observe that every word satisfying the constraints can be written as $x_1^{\alpha_1}\dots x_k^{\alpha_k}$ where $x_1 \cdots x_k$ is a cyclic Carlitz word and $1 \leq \alpha_i \leq l$. Thus, we can take the ordinary generating function of cyclic Carlitz words and substitute $z \mapsto q \sum_{i=1}^{l} x^i$ to obtain the ordinary generating function of the number of words satisfying the required constraint.
    \begin{align*}
        L_{q,l}(x) &= \frac{q^3 (q-1) \left(\sum_{i=1}^{l} x^i\right)^2}{\left(\sum_{i=1}^{l} x^i + 1\right)\left(1 - (q-1)\sum_{i=1}^{l} x^i\right)} \\
        &= q^2 (q-1) \left(\sum_{i=1}^{l} x^i\right)^2 \left(q^{-1} - (q^2 - q -1)\sum_{i=1}^{l} x^i\right)^{-1}.
    \end{align*}
    Applying the binomial theorem on $\left(q^{-1} - (q^2 - q -1)\sum_{i=1}^{l} x^i\right)^{-1}$ and rearrangement of the result obtain
    \begin{align*}
        L_{q,l}(x) =  \sum_{n=2}^{\infty} (q-1) q^{n} (q^2-q-1)^{n-2}  \left(\sum_{i=1}^{l} x^i\right)^{n}.
    \end{align*}
    After using the multinomial theorem, the generating function is rewritten as
    \begin{align*}
        L_{q,l}(x) =  \sum_{n=2}^{\infty} (q-1) q^{n} (q^2-q-1)^{n-2} \sum_{k_1 + \dots +k_l = n} \binom{n}{k_1,\dots,k_l} x^{\sum_{i=1}^{l-1} i k_i},
    \end{align*}
    and the coefficient $\left[x^n\right] L(x)_{q,l}$ is extracted to obtain the desired result.
\end{proof}

The run-length restricted primitive words are counted with the same procedure as the primitive $\epsilon$-balanced words, and the upper bound on the non-overlapping code with restricted run-lengths is determined the same way as $\epsilon$-balanced non-overlapping code, giving the following results.
\begin{lemma}
    The number of $q$-ary primitive words of length $n$ such that the length of no run in the word nor in any of its cyclic shifts exceeds $l$ equals
    \begin{align*}
        \sum_{d \mid n} \mu(d) \lvert L_{q,l}^d \rvert.
    \end{align*}
\end{lemma}

\begin{corollary}
    \begin{align*}
        S_l(q,n) \leq \sum_{d \mid n} \frac{\mu(d)}{n} \lvert L_{q,l}^d \rvert.
    \end{align*}
\end{corollary}

Now let us determine lower bounds on $S_l(q,n)$. If a non-overlapping code $C$ has a run of length $n-1$, then there exists $w \in C$ that starts with $n - 1$ identical symbols or ends with $n-1$ identical symbols. Let $L_1$ be the set of all letters $x \in \Sigma$ such that there exists a word in $C$ beginning in $x$, and $R_1$ the set of all letters $y \in \Sigma$ such that there exists a word in $C$ ending in $y$. Since no element in $L_1R_1$ occurs as both a prefix and a suffix in $C$, for any pair of symbols $x \in L_1$ and $y \in R_1$ at most one of the words $x^{n-1}y$ and $xy^{n-1}$ belongs to $C$. If we remove all such words, the remaining words constitute a non-overlapping code without any run of length $n-1$.
Therefore
\begin{align*}
    S_{n-2}(q,n) &\geq S(q,n) - \lvert L_1 \rvert \lvert R_1 \rvert \\
    &\geq S(q,n) - \lvert L_1 \rvert (q - \lvert L_1 \rvert) \\
    &\geq S(q,n) - q^2 / 4.
\end{align*}

\begin{proposition}
    For $q > 2$,
    \begin{align*}
        S_l(q,n) \geq (q-1)^{lm + 1 + r} (q-2)^{m},
    \end{align*}
    where $m$ is the quotient and $r$ the remainder of division of $n-2$ by $l$.
\end{proposition}

\begin{proof}
    Partition $\Sigma$ into $P$ and $S$ so that one symbol belongs to $P$ and $q-1$ symbols belong to $S$. The set $PS^{n-1}$ is non-overlapping by Construction~\ref{construction:fimmel}. To avoid repetitions longer than $l$ elements, we remove all words with the same letter on positions $1+kl$ and $2+kl$ for some positive integer $k$. The size of the resulting set $C$ equals
    \begin{align*}
        \lvert C \rvert &= \lvert P \rvert \lvert S \rvert^{l}(\lvert S \rvert -1)\lvert S \rvert^{l-1}(\lvert S \rvert -1)\lvert S \rvert^{l-1} \dots \\
        &= \lvert P \rvert \lvert S \lvert^{1 + (l-1)m} \left(\lvert S \rvert - 1\right)^m \lvert S \rvert^{r}.
    \end{align*}
    Since $q > 2$, $\lvert S \rvert - 1 > 0$. Hence, the code $C$ is non-empty, non-overlapping, and has no run of length $l$.
\end{proof}

Now, we proceed with constructions of maximal non-overlapping codes with restricted run-lengths. Observe step $i$ of the iterative construction of partitions $(L_j,R_j)_{j > 0}$.
If $w = uv$, $u\in L_j, v \in R_{i-j}$ for some $j$, $i-1 > j > 2$, then $u$ ends in a symbol from $R_1$ and $v$ starts in a symbol from $L_1$. Hence, $w$ has no runs longer than $l$ if and only if both $u$ and $v$ do not have them. Therefore, during step $i$, a subsequence of identical symbols can only be lengthened in the concatenations $L_1R_{i-1}$ and $L_{i-1}R_1$. In the first case the subset of words with run-length of at most $l$ equals 
\begin{align*}
X_i = \{xy \mid x \in L_1, y \in R_{i-1} \text{ starts with at most } (l-1) \\\text{ consecutive $x$'s}\},
\end{align*}
and in the second case the subset of words with run-length at most $l$ equals 
\begin{align*}
Y_i = \{yx \mid x \in R_1, y \in L_{i-1} \text{ ends with at most } (l-1) \\ \text{  consecutive $x$'s}\}.
\end{align*}

Therefore, Construction~\ref{construction:rll} (below) generates a set of non-overlapping codes with run-length of at most $l$. In particular, it creates the largest set of words with no run longer than $l$ contained in a code from $\mathcal{M}_{q,n}$. Hence, every maximal non-overlapping code with runs restricted by $l$ can be obtained from this construction. The generalisation, when run-length restrictions are different for distinct letters, is also straightforward.

\begin{construction}
    \label{construction:rll}
    Let the following hold for $n \geq 3$ and $l  < n-1$. \\
        (i) $(L_1, R_1)$ is a partition of $\Sigma$ into two non-empty parts, \\
        (ii) for $i \in \{2,\dots, l+1\}$, $(L_i, R_i)$ is a partition of $\bigcup_{j=1}^{i-1} L_j R_{i-j} $, \\
        (iii) and for every $i \in \{l+2,\dots, n-1\}$, $(L_i, R_i)$ is a partition of $X_i  \cup Y_i \cup \bigcup_{j=2}^{i-2} L_j R_{i-j}$.\\
    Then the code
    \[C = X_n  \cup Y_n \cup \bigcup_{j=2}^{n-2} L_j R_{n-j}\]
    is non-overlapping and its longest run does not exceed $l$.
\end{construction}

\section{Balanced non-overlapping codes with restricted run-lengths}\label{sec:balanced_rll}
To generate an $\epsilon$-balanced non-overlapping code with restricted run-lengths, we combine the observations from Constructions~\ref{proposition:mqntobqn} and \ref{construction:rll}. Instead of partitioning only the sets $L_i$ and $R_i$ into sets $L_i^j$ and $R_i^j$ based on the number of positive letters in a word as was done from $\epsilon$-balanced codes, we also partition the sets $X_i$ and $Y_i$ defined for codes with restricted run-lengths into sets $X_i^j$ and $Y_i^j$. 
Construction~\ref{construction:e-balanced-l-run-length} is an immediate consequence of the previous results.

\begin{construction}\label{construction:e-balanced-l-run-length}
Let $n\geq 3$. The code
\begin{align*}C = \bigcup_{j=j_{\text{min}}}^{j_{\text{max}}}  \left(X_i^j  \cup Y_i^j \cup \bigcup_{k=2}^{i-2} \bigcup_{m=\langle k+j-i \rangle} ^{\min(j,k)} L_k^mR_{i-k}^{j-m}\right).\end{align*}
where \\[2ex]

\begin{tabular}{p{.02\textwidth}p{.75\textwidth} }
    (i)& $\left(L_1^0, R_1^0\right)$ and $\left(L_1^1, R_1^1\right)$ are such partitions of $\Sigma_N$ and $\Sigma_P$, respectively, that $L_1^0 \cup L_1^1$ and $R_1^0 \cup R_1^1$ are both non-empty,\\
    (ii) &for $i \in \{2,\dots, l+1\}$ and $0 \leq j \leq i$,\\& $(L_i^j, R_i^j)$ is a partition of  $\bigcup_{k=1}^{i-1} \bigcup_{m=\langle k+j-i \rangle} ^{\min(j,k)} L_k^mR_{i-k}^{j-m}$, \\
    (iii) &for $i \in \{l+2,\dots, n\}$ 
    \begin{align*}
        X_i^j = \{xy \mid x \in L_1^0, y \in R_{i-1}^j \text{ starts with at most } (l-1) \\  \text{ consecutive $x$'s}\} \\
         \cup \{xy \mid x \in L_1^1, y \in R_{i-1}^{j-1} \text{ starts with at most } (l-1)\\  \text{ consecutive $x$'s}\},
    \end{align*} \\
    (iv) & for $i \in \{l+2,\dots, n\}$
    \begin{align*}
        Y_i^j = \{yx \mid x \in R_1^0, y \in L_{i-1}^j \text{ ends with at most } (l-1) \\ \text{  consecutive $x$'s}\} \\
        \cup \{yx \mid x \in R_1^1, y \in L_{i-1}^{j-1} \text{ ends with at most } (l-1) \\ \text{  consecutive $x$'s}\},
    \end{align*} \\
    (v) & for $i \in \{l+2,\dots, n-1\}$ and $0 \leq j \leq i$,\\ &$(L_i^j, R_i^j)$ is a partition of 
    \[X_i^j  \cup Y_i^j \cup \bigcup_{k=2}^{i-2} \bigcup_{m=\langle k+j-i \rangle} ^{\min(j,k)} L_k^mR_{i-k}^{j-m},\] \\
    (vi) &$j_{\text{min}} = \lceil (0.5 - \epsilon) n \rceil$ and $j_{\text{max}} = \lfloor (0.5 + \epsilon) n \rfloor$, \\
    \end{tabular} \\[2ex]

 \noindent is an $\epsilon$-balanced non-overlapping code without runs longer than $l$.
\end{construction}

Determination of the largest codes using this method is computationally exhaustive. Thus, we modify Constructions \ref{construction:bilotta} and \ref{construction:levy} to determine lower bounds on the size of a maximum polarity-balanced code that contains no runs longer than $l$. The comparison between the two constructions is given in Table~\ref{tab:rll}.

We define $\mathcal{D}_{2n,l}$ as the set of all balanced Dyck words of length $2n$ having no runs longer than $l$. Moreover, we define $\hat{\mathcal{D}}_{2n,l}$ to be the set of all balanced Dyck words of length $2n$ starting with a sequence of at most $l-1$ ones, ending in a sequence of at most $l-1$ zeroes and having no runs longer than $l$ in between. The modification of Construction \ref{construction:bilotta} is now stated using these restricted Dyck words.

\begin{construction}
    \label{construction:bilotta_rll}
    If $n$ is even, then the code \begin{align*}
        D_{2n+2,l} = \bigcup_{i\in\left[0, n/2\right]} \{a1b0: a \in \mathcal{D}_{2i,l}, b \in \hat{\mathcal{D}}_{2(n-i),l}\}
    \end{align*}
    is a non-overlapping code with no runs longer than $l$. If $n$ is odd, then the code 
    \begin{align*}
    D_{2n+2,l} = \bigcup_{i\in \left[0,(n+1)/2\right]} \{a1b0: a \in \mathcal{D}_{2i,l}, b \in \hat{\mathcal{D}}_{2(n-i),l}\} 
                \setminus \{1a'01b'0, a',b' \in \hat{\mathcal{D}}_{n-1,l}\}
    \end{align*}
    is a non-overlapping code with no runs longer than $l$.
\end{construction}

\begin{proof}
For an even $n$, the words in $D_{2n+2,l}$ are a subset of $D_{2n+2}$, hence, non-overlapping by Construction~\ref{construction:bilotta}. For an odd $n$, in addition, we have to observe that the words $1a'01b'0$ where $a'$ or $b'$ have run longer than $l$, start with $l$ ones or end in $l$ zeroes, are not members of the set $\bigcup_{i\in \left[0,(n+1)/2\right]} \{a1b0: a \in \mathcal{D}_{2i,l}, b \in \hat{\mathcal{D}}_{2(n-i),l}\}$. In particular, we are left to show that the runs in $a1b0$ are at most $l$. The last run of $a$ ends in 0, so it cannot be lengthened by the sequence $1b0$. On the other hand, the runs at both borders of $b$ are lengthened by one additional equal symbol. Hence, the longest run in $a1b0$ is at most $l$.
\end{proof}

As shown in Proposition~\ref{proposition:bilotta_rll_size}, one should first enumerate the restricted Dyck words to determine the size of the codes from Construction~\ref{construction:bilotta_rll}. To our knowledge, this problem has not been solved for general values of $l$. We provide some known results and simple observations in Propositions \ref{prop:restricted_dyck} and \ref{prop:d2n2}.

\begin{proposition}\label{proposition:bilotta_rll_size}
    \begin{align*}
        \lvert D_{2n+2,l} \rvert = \begin{cases}
            \sum_{i =0}^{n/2} \; \lvert \mathcal{D}_{2i,l} \rvert \lvert \hat{\mathcal{D}}_{2(n-i),l} \rvert & n \text{ even},\\
            \sum_{i =0}^{(n+1)/2} \; \lvert \mathcal{D}_{2i,l} \rvert \lvert \hat{\mathcal{D}}_{2(n-i),l} \rvert - \lvert \hat{\mathcal{D}}_{n-1,l}\rvert^2 & n \text{ odd}.
        \end{cases}
    \end{align*}
\end{proposition}

\begin{proof}
    We first observe that for every word in $a1b0$ where $a$ and $b$ are Dyck paths, there exist unique indices $i$ and $j$ such that $a \in \mathcal{D}_{i}$ and $b \in \hat{\mathcal{D}}_j$. 
    Moreover, for an odd $n$, every word $1a0$ such that $a \in \hat{\mathcal{D}}_{n-1,1}$ is also a word in $\mathcal{D}_{n+1,1}$. Hence, 
    \begin{align*}
        \{1a01b0, a,b \in \hat{\mathcal{D}}_{n-1,l}\} \subseteq \bigcup_{i\in \left[0,(n+1)/2\right]} \{a1b0: a \in \mathcal{D}_{2i,l}, b \in \hat{\mathcal{D}}_{2(n-i),l}\},
    \end{align*}
    and the formula follows.
\end{proof}

\begin{proposition}\label{prop:restricted_dyck}
\begin{align}
 \lvert \mathcal{D}_{2n,1} \rvert &= 1.\label{eq:4}\\
    \lvert \hat{\mathcal{D}}_{2n, 2} \rvert &= \lvert \mathcal{D}_{2n-4, 2} \rvert. \label{eq:1}\\
    \lvert \mathcal{D}_{2n,n-1} \rvert& = \lvert \mathcal{D}_{2n} \rvert - 1. \label{eq:2}\\
    \lvert {\mathcal{D}}_{2n,l-1} \rvert \leq \lvert &\hat{\mathcal{D}}_{2n,l} \rvert \leq  \lvert {\mathcal{D}}_{2n,l} \rvert. \label{eq:3}
\end{align}
\end{proposition}

\begin{proof}
    The only Dyck word containing no runs of length two is $(10)^n$; thus, Eq.~\eqref{eq:4}.
    To prove Eq.~\eqref{eq:1}, observe that a balanced Dyck word starting in a run of a single one and ending in a run of a single zero has a form $10b10$ where $b$ is a Dyck word of length $2n-4$. Eq.~\eqref{eq:2} holds since the only balanced Dyck word having a run of $n$ ones or zeroes is $1^n0^n$. We are left to show Eq.~\eqref{eq:3}.
    Every balanced Dyck word without runs of length $l$ is contained in $\hat{\mathcal{D}}_{2n,l}$. The words in $\hat{\mathcal{D}}_{2n,l}$ contain no runs longer then $l$. Thus $ {\mathcal{D}}_{2n,l-1} \subseteq \hat{\mathcal{D}}_{2n,l} \subseteq {\mathcal{D}}_{2n,l}$.
\end{proof}

\begin{proposition}\cite{donaghey:1980}
\label{prop:d2n2}
\begin{align*}
    \lvert \mathcal{D}_{2n,2} \rvert = \sum_{k=0}^{n-1} \frac{1}{n-k}\binom{n-k}{k+1} \binom{n-k}{k}.
\end{align*}
\end{proposition}

The modification of Construction~\ref{construction:levy} is also straightforward. For $l < k$, there are no words in $C(n,k)$ without zero runs of length $l+1$. For $l \geq k$, all zero runs in $C(n,k)$ satisfy the constraint, and we only need to constrain the run-lengths of ones.

\begin{construction}
    \label{construction:levy_rll}
    Let $n$ and $k$ be integers, $k \leq n/2$ and $l > k$. The set $C(n,k,l)$ that contains all words of the form $0^k1c1$ where $c$ is a binary word with exactly $n/2-2$ ones, no zero run of length $k$ and no run of ones having length $l$ is a polarity-balanced non-overlapping code with run-length of at most $l$.
\end{construction}

\begin{proposition}
    Let $n > 4$, $1 \leq k \leq n/2$ and $l \geq k$. The number of words in the code $C(n,k,l)$ obtained from Construction~\ref{construction:levy_rll} equals
    \begin{align*}
        \lvert C(n,k,l) \rvert &= \sum_{r=\lceil \frac{n}{2k}\rceil - 1}^{n/2-k} \sum_{i=-1}^{1} (2-\lvert i \rvert)\, c(n/2-k,r,k-1)\, c(n/2-2,r,l-1).
    \end{align*}
\end{proposition}

\begin{proof}
    The statement is obtained following the procedure from the proof of Proposition~\ref{proposition:size_levy} with an additional constraint on the sizes of $\beta_i$s being smaller than $l$.
\end{proof}

Note that $C(n,k,l)$ is not the largest subset of $C(n,k)$ that satisfies the required constraints. It could be expanded into a code $C$ by adding codewords of the form $0^k1c1$ where $c$ is a binary word with exactly $n/2-2$ ones, no zero runs of length $k$, starting with a zero run with length of at most $l-1$, ending with a run of ones with length of at most $l-1$, having some runs of $l$ ones, but no run of $l+1$ ones. Then, clearly $\lvert C(n,k,l) \rvert \leq \lvert C \rvert \leq \lvert C(n,k,l+1) \rvert$.

\begin{table}[ht]
    \centering\small
    \begin{tabular}{ccrrrrr}
    \hline
         $n$ & Construction & $l=2$ & $l=3$ & $l=4$ & $l=5$ & $l=6$\\\hline
         4 & \ref{construction:bilotta_rll} & 1 & 1 & 1 & 1 & 1 \\  
         & \ref{construction:levy_rll} & 1 & 1 & 1 & 1 & 1 \\
         6 & \ref{construction:bilotta_rll} & 2 & 3 & 3 & 3 & 3 \\  
         & \ref{construction:levy_rll} & 2 & 2 & 2 & 2 & 2 \\
         8 & \ref{construction:bilotta_rll} & 3 & 7 & 8 & 8 & 8 \\  
         & \ref{construction:levy_rll} & 2 & 3 & 3 & 3 & 3 \\
         10 & \ref{construction:bilotta_rll} & 5 & 16 & 22 & 23 & 23 \\  
         & \ref{construction:levy_rll} & 2 & 7 & 10 & 10 & 10 \\
         12 & \ref{construction:bilotta_rll} & 11 & 42 & 63 & 71 & 72 \\  
         & \ref{construction:levy_rll} & 2 & 18 & 28 & 30 & 30 \\
         14 & \ref{construction:bilotta_rll} & 20 & 109 & 183 & 216 & 226 \\  
         & \ref{construction:levy_rll} & 2 & 45 & 79 & 89 & 90 \\
         16 & \ref{construction:bilotta_rll} & 48 & 314 & 570 & 700 & 747\\  
         & \ref{construction:levy_rll} & 2 & 113 & 224 & 260 & 266 \\
         18 & \ref{construction:bilotta_rll} & 102 & 861 & 1,749 & 2,249 & 2,451 \\  
         & \ref{construction:levy_rll} & 2 & 385 & 634 & 756 & 782 \\
         20 & \ref{construction:bilotta_rll} & 247 & 2,591 & 5,681 & 7,570 & 8,400 \\  
         & \ref{construction:levy_rll} & 2 & 720  & 1,794 & 2,196 & 2,486 \\
         22 & \ref{construction:bilotta_rll} & 541 & 7,451 & 18,045 & 25,151 & 28,483 \\  
         & \ref{construction:levy_rll} & 2 & 1,823 & 5,076 & 6,375 & 8,740 \\
         24 & \ref{construction:bilotta_rll} & 1,324 & 22,836 & 60,155 & 86,771 & 99,932 \\  
         & \ref{construction:levy_rll} & 2 & 4,625 & 16,025 & 26,211 & 30,692 \\
         26 & \ref{construction:bilotta_rll} & 2,980 & 67,506 & 196,154 & 295,133 & 346,730 \\  
         & \ref{construction:levy_rll} & 2 & 11,754 & 52,984 & 90,604 & 107,728 \\
         28 & \ref{construction:bilotta_rll} & 7,353 & 210,014 & 664,401 & 1,035,980 & 1,236,959 \\  
         & \ref{construction:levy_rll} & 2 & 29,919 & 175,263 & 313,216 & 378,084 \\
         30 & \ref{construction:bilotta_rll} & 16,941 & 631,703 & 2,206,028 & 3,586,780 & 4,365,268 \\
         & \ref{construction:levy_rll} & 2 & 76,266 & 580,080 & 1,083,084 & 1,327,076 \\
         \hline
    \end{tabular}
    \caption{Comparison between the number of codewords in the largest polarity-balanced codes with restricted run-lengths obtained by Constructions \ref{construction:bilotta_rll} and \ref{construction:levy_rll}.}
    \label{tab:rll}
\end{table}

\section{Enumerative encoding}\label{sec:encoding}
The size of non-overlapping codes grows exponentially. Thus, holding the codewords in a table is not feasible for any practical application, and instead, an enumerative encoding is required. In Constructions~\ref{construction:fimmel}, \ref{proposition:mqntobqn} and \ref{construction:e-balanced-l-run-length}, the sizes of the partitions also grow exponentially, indicating that they should be equipped with a more space-efficient algorithmic procedure. We start by observing that if one of the sets $L_i$ and $R_i$ is empty in Construction~\ref{construction:fimmel}, knowing which one is empty is sufficient to compute both sets from the partitions with smaller indices. Without loss of generality, assume that $R_i$ is empty. We can order the elements in $L_i$ so that the $k_1$-th element $x$ lies in $L_jR_{i-j}$ if and only if 
\begin{align*}
    \sum_{l=1}^{j-1} \lvert L_l \rvert \lvert R_{i-l} \rvert < k_1 \leq \sum_{l=1}^{j} \lvert L_l \rvert \lvert R_{i-l}\rvert,
\end{align*}
and $x$ is the $k_2$-th element of $L_jR_{i-j}$ if $k_2 = k_1 - \sum_{l=1}^{j-1} \lvert L_l \rvert \lvert R_{i-l} \rvert$.
Assuming there exist orderings on the sets $L_j$ and $R_{i-j}$, $x$ is a concatenation of the $\lfloor k_2 / \lvert R_{i-j} \rvert\rfloor$-th element in $L_j$ and the $(k_2 \mod \lvert R_{i-j}\rvert)$-th element in $R_{i-j}$.

If the partitions $(L_l, R_l)_{l<i}$ and their sizes are stored ordered, the computation of $j$ requires $O(i^2)$ time, the computation of $k_2$ $O(i)$ and the computation and access of elements in $L_j$ and $R_{i-j}$ require constant time. The computation of $x$, therefore, requires $O(i^2)$ time. If, on the other hand, $R_j$ and $L_{i-j}$ are both empty, we could avoid storing the sets $L_j$ and $R_{i-j}$ and determine the prefix of $x$ in $L_j$ in $O(j^2)$ and the suffix in $R_{i-j}$ in $O((i-j)^2)$ using the same procedure as above. 
Hence, if we store the sizes of partitions $(L_l, R_l)_{l<n}$ and those partitions with $\lvert L_i \rvert \lvert R_i\rvert \neq 0$, we can compute the $k$-th word of $\bigcup_{i <n} L_iR_{n-i}$ in $O(n^3)$ since at every step the subword is either divided into two shorter parts in $O(n^2)$ or accesses from memory in $O(1)$, and, in the worst case, we repeat the division procedure $n-1$ times.

Now, let us explain that choosing a non-overlapping code with many partitions with empty parts is worth it. We have shown in \cite{stanovnik:2024} that among all non-overlapping codes $\bigcup_{i<n} L_iR_{n-i}$ with fixed partitions $(L_i, R_i)_{i\leq\lfloor n/2 \rfloor}$ that achieve the maximum size, there exists one that satisfies $\lvert L_i \rvert \lvert R_i \rvert = 0$ for all $i > n/2$.
Moreover, in the same article we determined the maximum non-overlapping codes for $q\in\{2,\dots,6\}$ and $n\in\{3,\dots,16\}$ and for all parameter values found a solution satisfying $\lvert L_i \rvert \lvert R_i \rvert = 0$ for all $i > 1+n(q+1)^{-1}$. Blackburn~\cite{blackburn:2015} also proved that whenever $n$ divides $q$ a non-overlapping of maximum size exists with $\lvert L_i \rvert \lvert R_i \rvert = 0$ for all $i > 1$.

The described algorithm can easily be adapted for $\epsilon$-balanced non-overlapping codes by ordering the words in $C$ so that the $k_1$-th element $x$ of $C$ is in $L_l^mR_{n-l}^{r-m}$ if and only if
\begin{align*}
    \sum_{i=1}^{l-1}\sum_{k=k_{\min}}^{k_{\max}} \sum_{j=k+i-n}^{\min(i,k)} \lvert L_i^j\rvert \lvert R_{n-i}^{k-j}\rvert &< k_1 \leq \sum_{i=1}^{l}\sum_{k=k_{\min}}^{k_{\max}} \sum_{j=k+i-n}^{\min(i,k)} \lvert L_i^j\rvert \lvert R_{n-i}^{k-j}\rvert, \\
    \sum_{k=k_{\min}}^{r-1} \sum_{j=k+i-n}^{\min(i,k)} \lvert L_l^j\rvert \lvert R_{n-l}^{k-j}\rvert &< k_2  \leq \sum_{k=k_{\min}}^{r} \sum_{j=k+i-n}^{\min(i,k)} \lvert L_l^j\rvert \lvert R_{n-l}^{k-j}\rvert, \\
    \sum_{j=k+i-n}^{m-1} \lvert L_l^j\rvert \lvert R_{n-l}^{r-j}\rvert &< k_3 \leq \sum_{j=k+i-n}^{m} \lvert L_l^j\rvert \lvert R_{n-l}^{r-j}\rvert,
\end{align*}
where 
\begin{align*}
k_2 \coloneqq k_1 - \sum_{i=1}^{l-1}\sum_{k=k_{\min}}^{k_{\max}} \sum_{j=k+i-n}^{\min(i,k)} \lvert L_i^j\rvert \lvert R_{n-i}^{k-j}\rvert
\end{align*}
and
\begin{align*}
k_3 \coloneqq k_2 - \sum_{k=k_{\min}}^{r-1} \sum_{j=k+i-n}^{\min(i,k)} \lvert L_l^j\rvert \lvert R_{n-l}^{k-j}\rvert.
\end{align*} The words in $L_i^j \cup R_i^j$ with $\lvert L_i^j \rvert \lvert R_i^j \rvert = 0$ should be ordered using the same method.
If we define an ordering on $\{(x^k, y^m) \mid x\in L_1, y \in R_1, 0 < k < l, 0 < m < l\}$ and proceed in the same manner as above, we also obtain an algorithm for non-overlapping codes without runs of length $l$.

\section{Conclusions}\label{sec:discussion}
This paper develops three methods that guarantee the generation of maximal non-overlapping codes satisfying the $\epsilon$-balance constraint, the run-length limit, and the inclusion of both constraints simultaneously. Moreover, we provide simpler algorithms to obtain polarity-balanced non-overlapping codes with restricted run-length. We also offer relations with combinatorial objects that should be enumerated to determine the corresponding code sizes.

We demonstrate that none of the existing methods for constructing polarity-balanced non-overlapping codes is optimal for all alphabets and block sizes. The only tested parameter value where Construction~\ref{construction:bilotta} fails to provide the optimum suggests that it might be optimal when $n/2-1$ is a prime power. To support this claim, the algorithm's search space used to determine the maximum size of Construction~\ref{proposition:mqntobqn} should be sufficiently reduced to enable computations for larger block sizes. In particular, we believe that it is possible to show that there exists a maximum $\epsilon$-balanced non-overlapping code satisfying $\lvert L_i^j \rvert \lvert R_i^j \rvert = 0$ for all $i > n/2$ and all valid choices of $j$. Moreover, for each $i > n/2$, there should exist a polynomial in $\{\lvert L_k^l \rvert, \lvert R_k^l \rvert \mid k \leq n- i, \; l \leq \lfloor 0.5+\epsilon) n\rfloor - j\}$ that determines which part of the partition $(L_i^j, R_i^j)$ is empty.

For the application in DNA-based storage, non-overlapping should have a large Hamming distance in addition to the constraints studied in this paper. Since our constructions generate maximal codes, the expected distance is small. If $x,z \in L_{i}$ and $y,w \in R_{j-i}$, the Hamming distance between $xy$ and $zw$ clearly equals $d_H(xy,zw) = d_H(x,z) + d_H(y,w)$. An open problem remains characterising pairs of words $x \in L_i R_{n-i}$ and $y \in L_jR_{n-j}$ at a large distance for $i \neq j$.

\section*{Declarations}
\subsection*{Funding}
This research was partially supported by the scientific research program P2-0359 and by the basic research project J1-50024, both financed by the Slovenian Research and Innovation Agency and by the infrastructure program ELIXIR-SI RI-SI-2 financed by the European Regional Development Fund, the Ministry of Science, Education and Sports and by the Slovenian Research and Innovation Agency.

\subsection*{Author Contribution}
\textbf{Lidija Stanovnik:} Conceptualization, Formal analysis, Writing - original draft.
\textbf{Miha Moškon:} Writing - review \& editing, Funding acquisition.
\textbf{Miha Mraz:} Supervision, Writing - review \& editing, Funding acquisition.

\subsection*{Competing Interests}
All authors certify that they have no affiliations with or involvement in any organisation or entity with any financial or non-financial interest in the subject matter or materials discussed in this manuscript.

\subsection*{Data availability}
Not applicable.

\end{document}